\documentclass[notitlepage,prb, twocolumn, superscriptaddress]{revtex4-1}
\usepackage{hyperref}
\usepackage[usenames,dvipsnames]{color}
\usepackage{amsmath}
\usepackage[english]{babel}
\usepackage{amssymb}
\usepackage{bm}
\usepackage{mwe}
\usepackage{xfrac}
\usepackage{graphicx}
\usepackage{epstopdf}
\usepackage{comment}
\usepackage{mathrsfs}
\usepackage{amsfonts}
\usepackage{bbold}
\usepackage{color,soul}
\usepackage{tabularx}
\usepackage{multirow}
\usepackage{mathtools}
\usepackage[]{todonotes}
\usepackage{algpseudocode}
\usepackage{algorithm}
\usepackage{booktabs}
\usepackage{titlesec}

\newcommand{\cdag}{c^{\dagger}}%
\newcommand{\sgn}{{\rm sgn}}
\newcommand{\GM}{\mathbf{G}}%

\renewcommand{\th}{\rm{th}}%

\usepackage{amsthm}
\newtheorem{theorem}{Theorem}[section]
\begin{document}

\title{Symbolic determinant construction of perturbative expansions}
\author{Ibsal Assi}
\email{iassi@mun.ca}
\affiliation{Department of Physics and Physical Oceanography, Memorial University of Newfoundland, St. John's, Newfoundland \& Labrador, Canada A1B 3X7} 
\author{J. P. F. LeBlanc}
\affiliation{Department of Physics and Physical Oceanography, Memorial University of Newfoundland, St. John's, Newfoundland \& Labrador, Canada A1B 3X7} 

\date{\today}
\begin{abstract}
We present a symbolic algorithm for treating perturbative expansions of Hamiltonians with general two-body interactions.  The method, formally equivalent to determinant Monte Carlo methods, merges well-known analytics with the recently developed symbolic integration tool, algorithmic Matsubara integration (AMI)  that allows for the evaluation of the imaginary frequency/time integrals. By explicitly performing Wick contractions at each order of the perturbative expansion we order-by-order construct the fully analytic solution of the Green's function and self energy expansions.  A key component of this process is the assignment of momentum/frequency conserving labels for each contraction that motivates us to present a fully symbolic Fourier transform procedure which accomplishes this feat. These solutions can be applied to a broad class of quantum chemistry problems and are valid at arbitrary temperatures and on both the real- and Matsubara-frequency axis. To demonstrate the utility of this approach, we present results for simple molecular systems as well as model lattice Hamiltonians.  We highlight the case of molecular problems where our results at each order are numerically exact with no stochastic uncertainty.
\end{abstract}

\maketitle

\section{Introduction}

Perturbation theories are a fundamental tool in a physicist's arsenal for tackling interacting electron systems. In many body perturbation theory (MBPT), physical observables are expressed as an infinite series where each subsequent order is represented by an exponentially large number of contractions generated from Wick's theorem. Each contraction requires the evaluation of integrals over the set of all internal variables. There are several ways to treat MBPT numerically the most popular perhaps being Diagrammatic Monte Carlo ($\rm DiagMC$) algorithms \cite{ProkofevDiagMC2008,vanhoucke:natphys,chen:2019,Kozik2010}.   
Standard DiagMC methods suffer from the fermionic sign problem that results from the large number of contractions (diagrams) with alternating sign.\cite{ProkofevDiagMC2008,vanhoucke:natphys} In recent years, determinant methods have been introduced that can somewhat mitigate this issue.\cite{Rubtsov2005,Burovski2006,EGull2011} The connected determinant diagrammatic Monte Carlo (CDet) method was introduced to treat perturbative expansions and avoids the factorial scaling of diagrams at exponential cost\cite{RossiCDet2017,Jia2022,FedorFPM2022}  

Those methods, however, are based on the Matsubara formalism for finite temperatures and require numerical forms of analytic continuation in order to produce dynamical properties in real-frequency or real-time.  More recently the advent of algorithmic Matsubara integration (AMI)\cite{AMI,libami} method allows us to symbolically evaluate  summations over Matsubara frequencies and has been successfully applied to a number of physical problems such as the 2D Hubbard model\cite{mcniven:2021,mcniven:2022,taheri:2020,burke:2023} as well as the uniform electron gas.\cite{igor:spectral,leblanc:2022} AMI provides access to real frequency calculations via textbook analytic continuation, the replacement $i\omega_n \to \omega +i0^+$, which avoids ill-posed numerical analytic continuation schemes.\cite{Levy2016} It reduces the sampling space of internal variables minimizing the effect of the curse of dimensionality and reducing overall numerical uncertainty.

In this work, we build on determinental methods by introducing a fully algorithmic approach which we call the symbolic determinant method (symDET).  By combining the determinant approach with AMI we can apply MBPT to extremely general Hamiltonians relevant to quantum chemistry and condensed matter physics.  We start by generating Wick contractions symbolically and then proceed to Fourier transform those contractions also symbolically.  We then perform the integrals over the internal variables with the use of AMI for evaluating the Matsubara summations. In the next section, we introduce the elements of this algorithm in detail.  We then provide several applications in the following section, and provide a summary.

\section{Model and Methods}
\subsection{Two-Body Hamiltonian}

We discuss the evaluation of a very general two-body Hamiltonian
with two terms; a single-particle term, $H_0$, and a generalized four-operator interaction term, $H_V$. These are given by 
\begin{equation}
H =
\underbrace{\sum_{ab} h_{ab}\cdag_a c_b}_{H_0} +
\underbrace{\frac{1}{2} \sum_{abcd} U_{abcd} \cdag_a \cdag_c c_d c_b}_{H_V}.\vspace{10pt}
    \label{eq:h}
\end{equation}
Here $a$ and $b$ are arbitrary band indices - that might also include momenta or spin degrees of freedom - and the $c_i^\dagger$ and $c_i$ represent standard creation and annihilation operators in the state $i$, respectively, and the values of $h_{ab}$ represent one-electron integrals while $U_{abcd}$ is the two-electron interaction matrix. 
 The presumption for finding solutions to the model are that the single-particle term, $H_0$, is known and diagonal allowing us to perform an expansion in powers of the interaction term.

\subsection{Perturbative expansion of Green's function}\label{sec:IIb}

We define the non-interacting Green's function
\begin{equation}\label{eq:niG}
    G_{ba}^{0}(\tau)=-\langle c_b(\tau)c_a^\dagger(0) \rangle = \left[ (-\partial_\tau +\mu) \mathbb{1} -h \right]^{-1}_{ba},
\end{equation}
here written in imaginary time, $\tau$.  Later we will perform the Fourier transform to represent the Green's function for Matsubara frequency, $i\nu_n$. In general, $h_{ab}$ may not be diagonal which results in a non-diagonal Green's function. Without loss of generality we simplify the problem by presuming that $H_0$ can be represented in a diagonal basis and that the interaction $U_{abcd}$ is known in that diagonal basis. Thus, we can rewrite the diagonal Green's function on the Matsubara axis as
\begin{equation}
    G_{ab}^{(0)}(i\nu_n)=\frac{\delta_{ab}}{i\nu_n-h_{ab}},\label{eq:g}
\end{equation} 
where $\delta_{ab}$ is the Kronecker delta. In this representation the poles of the Green's function can be symbolically determined and this is necessary when implementing the AMI method.\cite{AMI} If $h_{ab}$ is not diagonal, then $G$ is not diagonal and the pole structure of the Green's function becomes obfuscated by the matrix inversion process.

With the target of generating the order-by-order expansion of $H_V$ we start by following the standard construction of the $m$th order correction to the imaginary time Green's function as
\begin{equation}\label{eq:sigma}
   G_{ba}^{(m)}(\tau)=\frac{(-1)^m}{ m!}\Bigg\langle \mathcal{T}\Bigg[\displaystyle\prod_{\ell=1}^{m}\int_{0}^{\beta}d\tau_{\ell}H_{V}(\tau_{\ell})\Bigg]c_{b}(\tau) \cdag_{a}(0) \Bigg\rangle_0  
\end{equation}
where $\mathcal{T}$ is the time ordering operator and $\beta=T^{-1}$ is the inverse temperature in units of the Boltzman constant $k_B$. We see that at order $m$ we must compute the expectation value of a sequence of $4m$ creation and annihilation operators attached to times $\tau_m$, in addition to the external operators $c_{b}(\tau)$ and $\cdag_{a}(0)$.

This expectation value can be evaluated using Wick's theorem, replacing the expectation value with a sum of all possible contractions of creation and annihilation operators.  This is typically accomplished in matrix form with rows and columns represented by annihilation and creation operators, respectively.  One can then generate all possible contractions - while also keeping correct track of the fermionic sign arising from commuting fermionic operators - by just taking the determinant of said matrix.\cite{RossiCDet2017,FedorFPM2022}

For this we define $\GM$ to be a $(2m+1)\times (2m+1)$ matrix in which the
rows (columns) correspond to the $2m$ annihilation (creation) operators plus an additional entry in each for the external vertices.
We introduce column and row indices $\alpha,\beta$ such that
\begin{equation}
\begin{split}
    &\{a_\alpha\} := \{a_1, c_1, a_2, c_2, \ldots, a_m, c_m, a_{out}\},\\
    &\{b_\beta\} := \{b_1, d_1, b_2, d_2, \ldots, b_m, d_m, b_{in}\},\\
\end{split}
\label{eq:colrowidx}
\end{equation}
and define the matrix elements\cite{Jia2022}
\begin{equation}
    \GM_{\beta\alpha} := -\langle c_{b_\beta}(\tau_\beta) \cdag_{a_\alpha}(\tau_\alpha) \rangle_0
    = g_{b_\beta a_\alpha}(\tau_\beta-\tau_\alpha + 0^-)=g_{\alpha \beta}.
    \label{eq:Gmatelem}
\end{equation}
The full matrix can then be written
\begin{equation}
    \GM := \left[\begin{matrix}
        g_{11}     & {g}_{12}   & \cdots    & {g}_{1n}\\
        g_{21}     & {g}_{22}   & \cdots    & {g}_{2n}\\
        \vdots          & \vdots        & \ddots    & \vdots\\
        {g}_{n1}     & {g}_{n2}   & \cdots    & {g}_{nn}
    \end{matrix}\right],
    \label{eq:Gmat}
\end{equation}
where $n=2m+1$.

This construction has been presented numerous times and forms the basis for determinant Monte Carlo methods applied to many-body systems.\cite{fedor:2020,Jia2022,rossi2017determinant,RossiCDet2017,FedorFPM2022}
In the standard prescription, the $\GM$ matrix is populated in the realspace and imaginary-time. The determinant procedure is typically evaluated numerically by inserting numerical values for the imaginary time Green's function, and sampling over all continuous times $\tau_m$.  There is one caveat to doing this is that the terms generated represent both connected and disconnected Feynman graphs.  Removing the disconnected components can be accomplished with the recent method described by Rossi et al.\cite{rossi2017determinant}.

\subsection{Algorithmic Matsubara Integration}

The method of algorithmic Matsubara integration, introduced in Ref.~\onlinecite{AMI}, was presented as a general procedure for the analytic evaluation of the temporal integrals of arbitrary Feynman diagram expansions. 
In essence, AMI is a straightforward application of residue theorem that stores the minimal information required to construct the analytic solution for an arbitrarily complex integrand comprised of a product of bare Green's functions.  The Matsubara integrals are not conceptually challenging to perform and are the topic of numerous textbook exercises.  The difficulty in performing those contour integrals lies only in that the number of poles and number of resulting analytic terms grows exponentially with diagram order.

Using the existing AMI library\cite{libami} the result of AMI is stored in three nested arrays: Signs/prefactors $S$, complex poles $P$, and Green's functions $R$. From these three objects, whose storage is quite minimal, one can then construct the analytic expression symbolically through elementary algebraic operations \cite{AMI}.

The beauty of such a result is that the analytic expression is analytic in external variables, allowing for true analytic continuation of $i\nu_n \to \nu +i0^+$, and is also an explicit function of temperature, $T$. Further, for a given graph topology the AMI procedure need only be performed once and is valid for any choice of dispersion in any dimensionality and can be applied to model systems for a wide variety of Feynman diagrammatic expansions.\cite{leblanc:2022,taheri:2020,GIT,burke:2023, farid:2023}
For the present work, we use AMI as a method for evaluating Matsubara integrands and the determinant construction replaces the usual Feynman diagram representation.   

\subsection{Symbolic determinant method - symDET}

In this work we deviate from the standard determinantal scheme mentioned in Sec.~\ref{sec:IIb}.  Here we will outline a procedure to generate the perturbative expansion in terms of bare propagators such that the integrands of Eq.~\ref{eq:sigma} are in a form suitable for AMI.\cite{AMI}  AMI operates in the energy(momentum)/frequency basis and cannot be applied to imaginary or real-time Green's functions - though there exist non-algorithmic variants designed in the same spirit that may perhaps overcome this barrier.\cite{jaksa:2020}
Each term in Eq.~\ref{eq:sigma} is represented as a function of a set of imaginary times.  To translate these to a form amenable to AMI we require tools to:
\begin{enumerate}
    \item Perform the symbolic Wick's contractions for each term in Eq.~(\ref{eq:sigma}).
    \item Identify and remove disconnected topologies. 
    \item Perform the nested sequence of Fourier transform from $\tau\to i\nu_n$ symbolically.
\end{enumerate}
We provide the solution to each issue in the following subsections.

\subsubsection{Symbolic Wick's contractions}
When creating a symbolic representation of the matrix form of Eq.~(\ref{eq:Gmat}) each element with row and column indices $\alpha$ and $\beta$ is just a function of those indices.  We can therefore generate a symbolic representation by replacing the entries  with their row and column indices, $\GM_{\alpha \beta}\to (\alpha,\beta)$. 

If we can take a determinant of this matrix and store each term separately, we will have generated the expressions that represent the $n!$ connected and disconnected diagrams. 
Evaluation of numerical determinants can be accomplished in $O(n^3)$ time, an advantage of modern determinantal methods\cite{Rossi:shiftedaction,rossi:2018,rossi2017determinant} but since we want to proceed symbolically there is no obvious route to such fast evaluations. Instead, we take the most pedantic approach and simply store the explicit parameters of each term in the determinant.  While this factorial scaling sounds problematic the tradeoff is an analytic expression that is exact to machine precision.  This is in lieu of stochastic methods that, while they can evaluate determinants quickly, must perform temporal integrals via Monte-Carlo sampling, a process that for high accuracy requires typically $10^6\to 10^8$ samples. We expect that for low orders we will arrive at a precise numerical result with fewer operations despite this factorial scaling. 

To proceed we use the Leibniz formula for an $n\times n$ matrix, $A$ with elements $a_{i,j}$:

\begin{equation}
    \det(A)=\sum_{p \in \mathcal{P}_n} \left( \sgn(p) \prod_{i=1}^n a_{i,p_i}   \right).
    \label{eq:det}
\end{equation}
In this expression, $p=(p_1,p_2,\cdots,p_n)$ is a permutation of the set $\{1,2, \dots,n \}$ and $\mathcal{P}_n$ is the set of all such permutations. $\sgn(p)$ is the signature of $p$ defined as +1 whenever the reordering requires an even number of interchanges and -1 when an odd number is required. 
Finding the permutations of $p$ and the associated signs is a straightforward computational problem.  To do this symbolically we generate a permutation $p$ and then store the indices of $a_{i,p_i}=(i,p_i)$ for each $i$. Each term in Eq.~\ref{eq:det} is then completely defined by a vector of such pairs, and a single +1/-1 sign prefactor.

This represents a major departure from typical determinantal QMC methods\cite{hirata:2015,ferrero:2018}  where such a matrix is filled with numerical values.  In our case we have yet to assign values to the entries and instead we want to store the information required to later symbolically construct the expression.  

\subsubsection{Two in one: The Symbolic Fourier Transform}
A very interesting and useful property of Feynman diagrams is that the set of possible diagram topologies is independent of coordinate and temporal labelling of each vertex. However, in the contractions of Eq.~\ref{eq:Gmat} each topology may appear multiple times - as is famously the case for a single-band problem where the $m!$ denominator is precisely cancelled by $m!$ duplicates of each topology.  
Since we have each contraction - we are free to represent each as a graph in momentum and Matsubara frequency space.  However, in doing so one would need to develop an internally consistent labelling of each graph - a process that is fundamentally non-local in diagram topology and also is not unique. 

Instead we choose to mimic the analytic process and have devised an analytic representation of the temporal Fourier transform. The procedure, detailed in Appendix~\ref{app:SFT}, sorts the contraction pairs $(i,j)$ that represent imaginary time Green's functions spanning between times $\tau_{\lfloor{\frac{i}{2}}\rfloor}$ and $\tau_{\lfloor{\frac{j}{2}}\rfloor}$.
The pairs are then separated into three lists $A$, $B$, and $C$. 
Since the contraction pairs are effectively source/target sets the connectivity of the contraction can be determined directly as is done in graph theory, identical to a depth first search, at minimal expense, scaling with the number of vertices, $n_v$, which is typically small and scales as $2^n$ for perturbation order $n$.  If at the end of the process the number of pairs in $A$ is $n-1$, in $B$ is $2$ and in $C$ we have $n$ pairs, then the diagram is connected. Now, the symbolic Fourier transform of the time integrals is done by simply converting those three lists to matrices as described in Appendix \ref{app:SFT}.  The advantage of this is that one obtains an unique set of internal labels that obey energy and momentum conservation at all vertices. The main result is given in Eq.~\ref{eqn:AMIfreq} which is a matrix with entries zero, and $\pm1$.

\subsection{Evaluation}

At this stage, our $n^{\rm th}$ order perturbative expansion is of the form
\begin{equation*}
    G_{ba}^{(n)}=\frac{(-1)^n}{n!}g_{b}(i\omega_{\rm ex})g_{a}(i\omega_{\rm ex})\times
\end{equation*}
\begin{equation}
    \sum_{e_1,\dots,e_{2n-1}}\sum_{\{\Omega_n\}}\sum_{c\in \mathcal{C}}\prod_{j}^{2n-1}g^{j}_{e_j}(\boldsymbol{\alpha_j}\cdot\boldsymbol{\omega})
\end{equation}
where the first summation is over  the internal variables,$e_i$, (e.g. orbital numbers, momenta, spin, or a mix of them etc), the second summation is over the set of internal Matsubara frequencies, and the last is over all contractions belong to the set $\mathcal{C}$. Here $\boldsymbol{\alpha_j}$ is the $j^{\rm th}$ row in Eq. \ref{eqn:AMIfreq}, $\boldsymbol{\omega}=(\Omega_1,\dots,\Omega_n,\omega_{\rm ex})^T$ and 
\begin{equation}
    g^{j}_{e_j}(\boldsymbol{\alpha_j}\cdot\boldsymbol{\omega})=\frac{1}{\boldsymbol{\alpha_j}\cdot\boldsymbol{\omega}-\varepsilon_{e_j}}
\end{equation}
is the Fourier transformed free propagator. In the case of molecular problems, or generically \textit{discrete} systems, one performs the $e_i$ summations directly such that our algorithm gives the exact value of the perturbative expansion. However, in the case of lattice problems, we use stochastic sampling over momenta, we obtain results with stochastic error-bars. In both cases, the Matsubara summations are evaluated exactly.     

\section{Applications}
\subsection{Application to Molecular Chemistry - $H_2$}
Molecular hydrogen is the simplest system to consider as a test-bed for method development and here we start with the simplest representation in the STO-6g basis which describes the interaction between the two hydrogen atoms having only 1s orbitals.  
In particular, we will see later in Section~\ref{sec:IIIc} that the two state problem is the basic component of a single-band with spin $\uparrow$/$\downarrow$ and therefore correct results for the STO-6g basis are paramount in developing the method beyond simple problems.
 We use the pyscf package\cite{pyscf} to obtain the Hartree-Fock solutions for the STO-6g basis from which we compute the self energy on the Matsubara axis illustrated in Fig. \ref{fig:H2imagfreq}.  We have compared our results in detail to those in Ref.~\cite{Jia2022} and find that our exact result is within stochastic error bars of that work.  
 Different from their result, our starting eigenstates are asymmetic resulting in distinct values of $\Sigma_{00}$ and $\Sigma_{11}$ while the off-diagonal self energy terms are zero in this case. 
 While we stop at fourth order, there is no conceptual hurdle to evaluating higher orders or larger basis sets.  However, the computational expense is factorial in order and exponential in basis.  Nevertheless, the procedure is easily parallelizeable. 
 
 The real advantage to our approach is the direct evaluation of real frequency properties.  By symbolically replacing $i\omega_n \to \omega+i\Gamma$ we can plot the self energy in real frequencies shown in Fig.~\ref{fig:H2realSig} for a particular choice of $\Gamma$ that can be made arbitrarily small.  Here we focus on a relevant frequency range where there is an expected new peak that is created by a sharp feature in ${\rm Re}\Sigma(\omega)$ such that the interacting Green's function gains one or more additional poles.  This is seen in the spectral function as shown in Fig.~\ref{fig:H2SF}. The dominant peaks remain those of the non-interacting dispersion while additional peaks - shown in the insets - appear at energies offset by the peak difference $\Delta E=h_{11}-h_{00}$ which is expected based on the second order expansion.  At fourth order shown, there are two additional poles instead of a single peak near $\omega=\pm 2$.  

As an example for a larger basis set, we compute the self-energy for $H_2$ in the 10 orbital cc-pVDZ basis representation  as shown Fig.\ref{fig:H10imagfreq}. This basis is five times larger than its STO-6g counter part, stressing our ability to study larger molecules with symDET. 
  
An interesting implication of these calculations is the ability to perform self-consistent perturbation theory beyond the well-known GF2 method. For molecular chemistry problems, this implementation of ${\rm GF}_n$ is exact at each order and is valid at finite or zero temperatures, and at any physical parameters. For example, the binding energy for molecules is obtained by generating the poles of the full propagator which is easy via the AMI  part of our code.\cite{hirata:2015}  

\begin{figure}
\label{fig:H2}
\centering
    \includegraphics[width=1.0\linewidth]{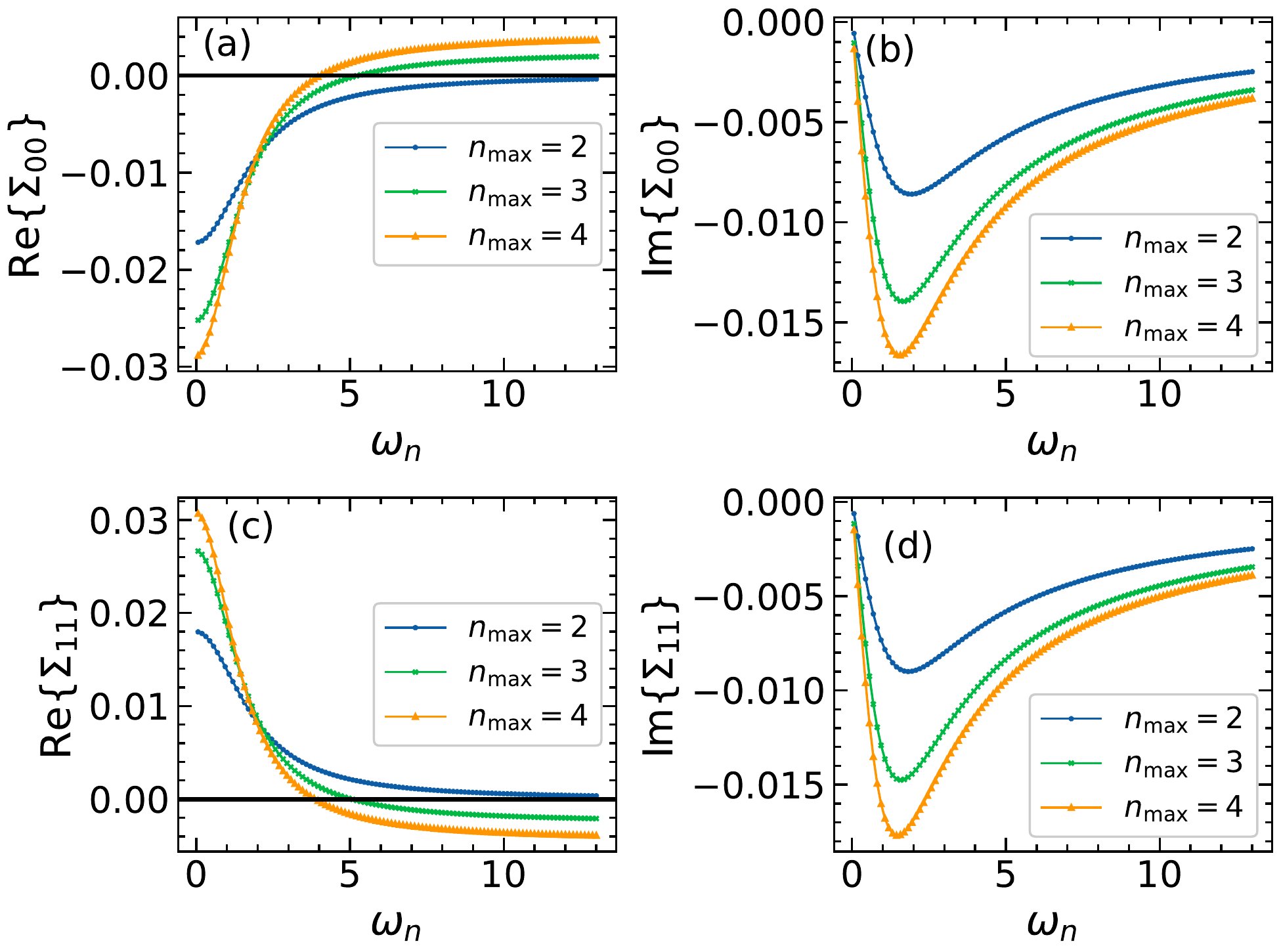}
\caption{(a) \& (b) The real and imaginary parts of the self-energy for ${\rm H}_2$ in the STO-6g basis with external band indices $a_{{\rm ex}}=b_{{\rm ex}}=0$. (c) \& (d) plots of the real and imaginary parts of $\Sigma^{{\rm H}_2}$ for $a_{{\rm ex}}=b_{{\rm ex}}=1$. Here we took $\beta=50.0$.}
\label{fig:H2imagfreq}
\end{figure}

\begin{figure}
\label{fig:H2}
\centering
    \includegraphics[width=1.0\linewidth]{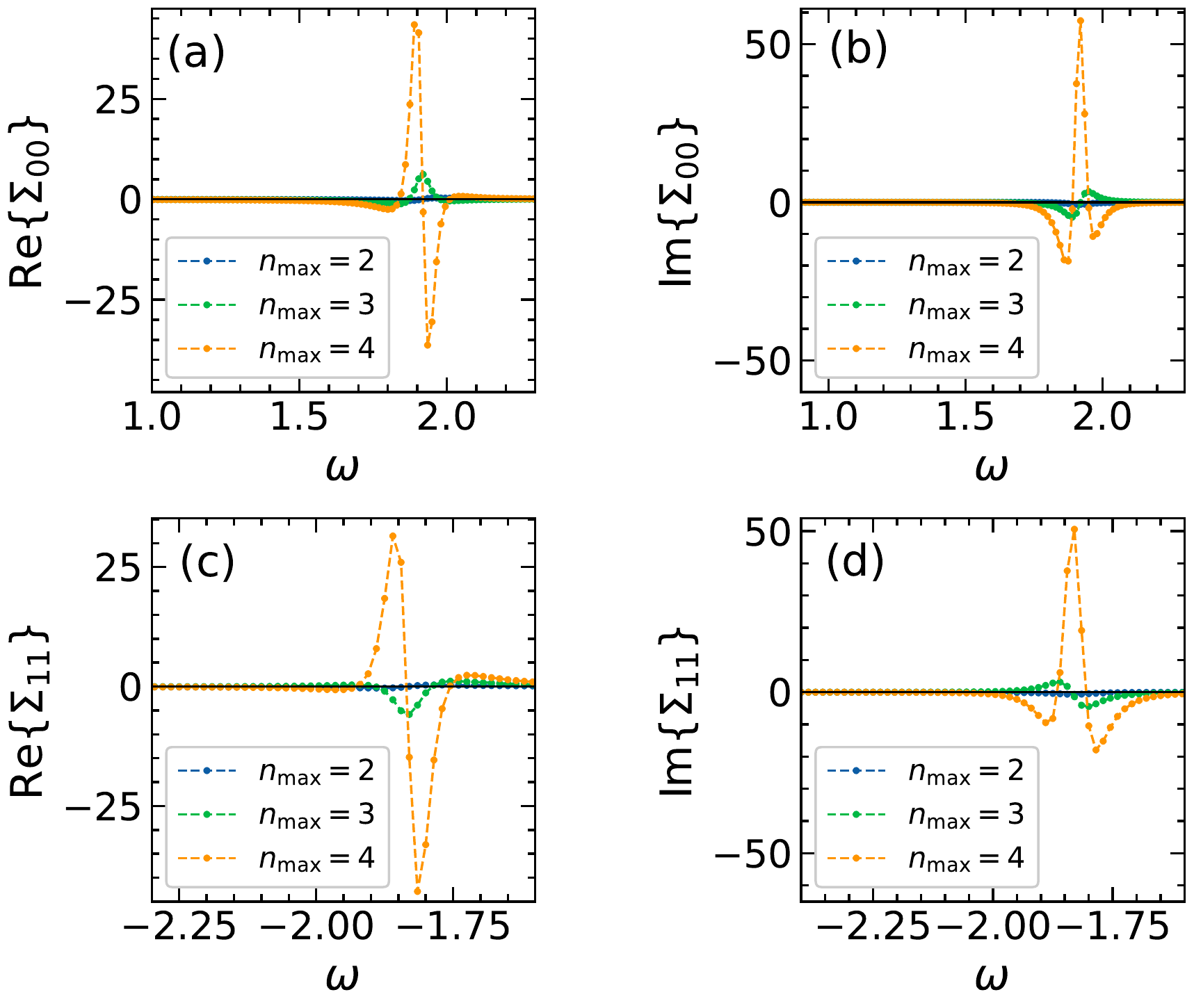}
\caption{(a) \& (b) the real and imaginary parts of $\Sigma_{00}$, while (c) \& (d) are the components of the self-energy for the second band for $H_2$ (in the STO-6g basis) on the real frequency axes. Here we took the regulator $\Gamma=0.05$}
\label{fig:H2realSig}
\end{figure}

\begin{figure}
\label{fig:H2}
\centering
    \includegraphics[width=0.9\linewidth]{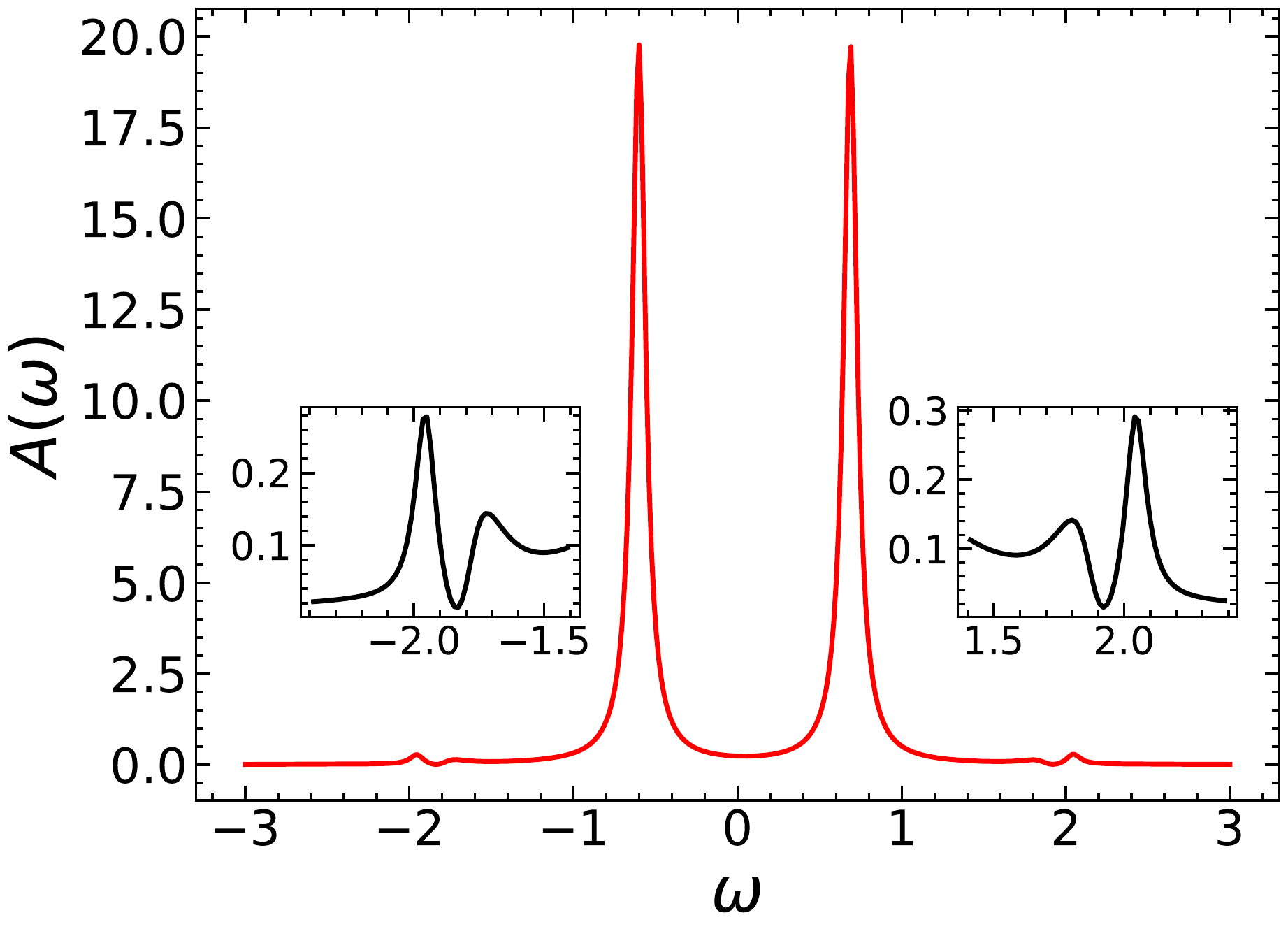}
\caption{The spectral function for ${\rm H}_2$ in the STO-6g basis truncated at 4th order. In inset data is zoom out of the extra peaks with lower intensity. Here we took the regulator 0.05. }
\label{fig:H2SF}
\end{figure}

\begin{figure}
\label{fig:H10}
\centering
    \includegraphics[width=1.0\linewidth]{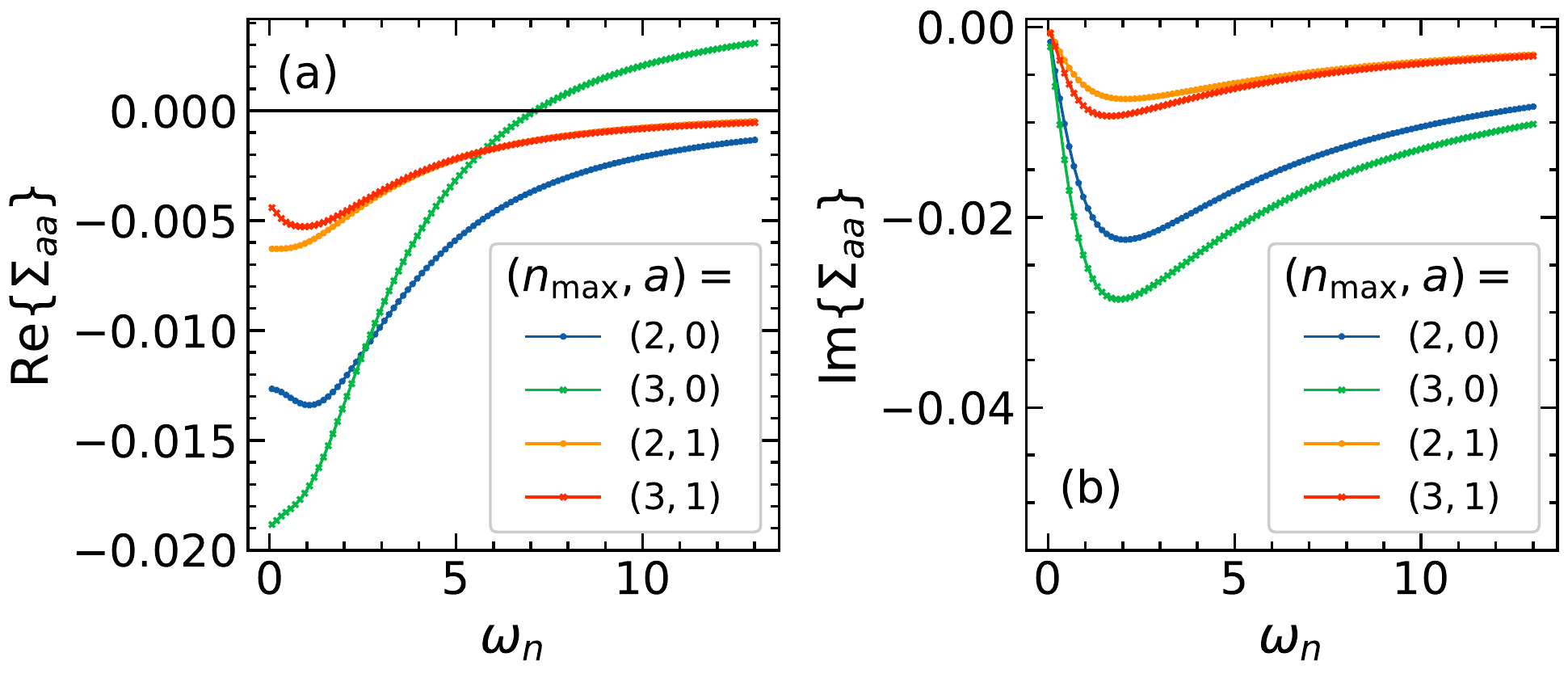}
\caption{The Matsubara self-energy for Hydrogen in the cc-pVDZ basis versus the Matsubara frequency with $\beta=50.0$. (a) The real part of the self-energy components $(0,0)$ and $(1,1)$ truncated at second ($n_{\rm max}=2$) and third ($n_{\rm max}=3$) orders, and (b) are the imaginary counterparts.} 
\label{fig:H10imagfreq}
\end{figure}

\subsection{The Hubbard-Dimer Model}
To demonstrate the versatility of our approach we study the Hubbard dimer. 
The model consists of two sites each has a spin $1/2$ particle. The model we use is described below \cite{Dimer_EG}
\begin{equation}\label{eq:hd}
    H=H_0+H_U+H_H+H_{SB}-\mu\sum_{i,\sigma} c_{i\sigma}^{\dagger}c_{i\sigma}
\end{equation}
where $H_0=-t\sum_{\sigma=\uparrow,\downarrow}(c_{0\sigma}^{\dagger}c_{1\sigma}+c_{0\sigma}c_{1\sigma}^{\dagger})$, is the hoping term for electrons between the two sites,   $H_U=U\sum_i n_{i\uparrow}n_{i\downarrow}-\frac{U}{2}\sum_{i\sigma}n_{i\sigma}$ describes the onsite interaction, $H_{H}=H\sum_{i}(n_{i\uparrow}-n_{i\downarrow})$ the interaction due to an applied magnetic field, and $H_{SB}=U_a (n_{0\uparrow}n_{0\downarrow}-n_{1\uparrow}n_{1\downarrow})+\mu_a (n_{0\uparrow}+n_{0\downarrow}-n_{1\uparrow}-n_{1\downarrow})+H_a(n_{0\uparrow}-n_{0\downarrow}-n_{1\uparrow}+n_{1\downarrow})$ is a symmetry-breaking term. By diagonalizing the quadratic part of the full Hamiltonian, we can rewrite the above Hamiltonian in the usual form
\begin{equation}
    H=\sum_{a=1}^{4}\varepsilon_a f_a^\dagger f_a+\frac{1}{2}\sum_{abcd}V_{abcd}f^\dagger_a f^\dagger_c f_d f_b
\end{equation}
where $\varepsilon_a$ is the effective dispersion, $f_a^\dagger$ ($f_a$) are the creation(annihilation) fermionic operator, and $V_{abcd}$ is the effective interaction, where both $\varepsilon_a$ and $V_{abcd}$ can be obtained analytically for this four-band system. In this example, the self-energy in this basis is not diagonal (rather a block-diagonal). As an illustration, we plot the imaginary and real parts of $\Sigma_{00}$ and $\Sigma_{01}$ up to fourth order for $U=2.5t=5.0$, $\mu=0.7$, $H=0.30$, $U_a=0.5$, $\mu_a=0.20$, $H_a=0.030$ and $\beta=2.0$.   

\begin{figure}
\centering
    \includegraphics[width=1\linewidth]{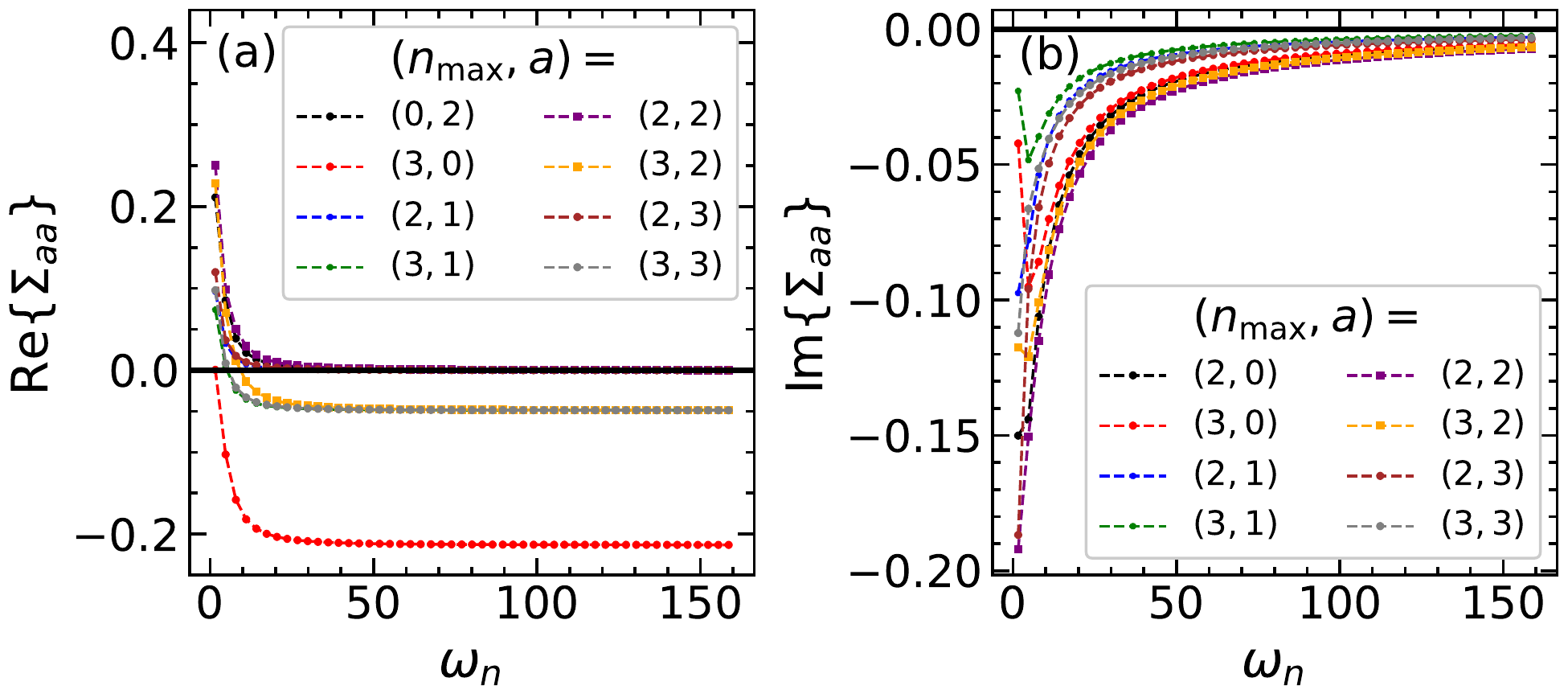}
\caption{(a)\& (b) the real and imaginary parts of the diagonal elements of the self-energy matrix for $n_{\rm max}=2,3$ for the Hubbard Dimer model. Here we took $t=1.0$, $U=2.5$, $\mu=0.70$, $H=0.30$, $U_a=0.50$, $\mu_a=0.20$, $H_a=0.030$, and $\beta=2.0$.}
\label{fig:Dimer}
\end{figure}

\subsection{Single-Band Hubbard model}\label{sec:IIIc}

The simplest starting point for considering a lattice Hamiltonian is the single-band Hubbard model of spin-$\sfrac{1}{2}$ Fermions on a square lattice.  The model is typically written in real-space notation as
\begin{eqnarray}\label{E:Hubbard}
H = \sum_{\langle ij\rangle \sigma} t_{ij}c_{i\sigma}^\dagger c_{j\sigma} + U\sum_{i} n_{i\uparrow} n_{i\downarrow},
\end{eqnarray}
where $t_{ij}$ is the hopping amplitude, $c_{i\sigma}^{(\dagger)}$ is the annihilation (creation) operator at site $i$, $\sigma \in \{\uparrow,\downarrow\}$ is the spin, $U$ is the onsite Hubbard interaction, $n_{i\sigma} = c_{i\sigma}^{\dagger}c_{i\sigma}$ is the number operator,  $\mu$ is the chemical potential, and $\langle ij \rangle$ restricts the sum to nearest neighbors. For a 2D square lattice we take $t_{ij}=-t$, resulting in the free particle energy 
\begin{eqnarray}\label{E:dispersion}
\epsilon(\textbf k)=-2t[\cos(k_x)+\cos(k_y)]-\mu.
\end{eqnarray}

Mapping this problem to Eq.~\ref{eq:h} leads to an effective problem of two degenerate bands with states $\uparrow=(k,\sigma=\uparrow)$ and $\downarrow=(k,\sigma=\downarrow)$ and the band indices are then summed over up and down basis.  This leads to a diagonal and spin independent $h_{ab}=\epsilon_k \delta_{ab}$ and an interaction term independent of momentum with entries $U_{\uparrow \uparrow \downarrow \downarrow}=U_{\downarrow \downarrow \uparrow \uparrow}=U$ and all other $U$ elements are zero. 

Due to the additional $k$-indices, after processing with AMI each $m^{\rm th}$ order wick contraction contains an $m$-dimensional integral over internal momentum vectors which requires approximate numerical integration methods to evaluate. Otherwise the procedure is unchanged from the two-band case of $H_2$ in the STO-6g basis which highlights the importance of that problem as a benchmark.
As an illustration, we have calculated the self-energy for the 2D square lattice on the Matsubara axis shown in Fig.~\ref{fig:Hubb_matsubara} for doped cases $\mu \neq 0$. Moreover, the exact same expressions can be used to generate the matching real-frequency results which we show in Fig \ref{fig:Hubb_mu}. 

\begin{figure}
\centering
    \includegraphics[width=1.0\linewidth]{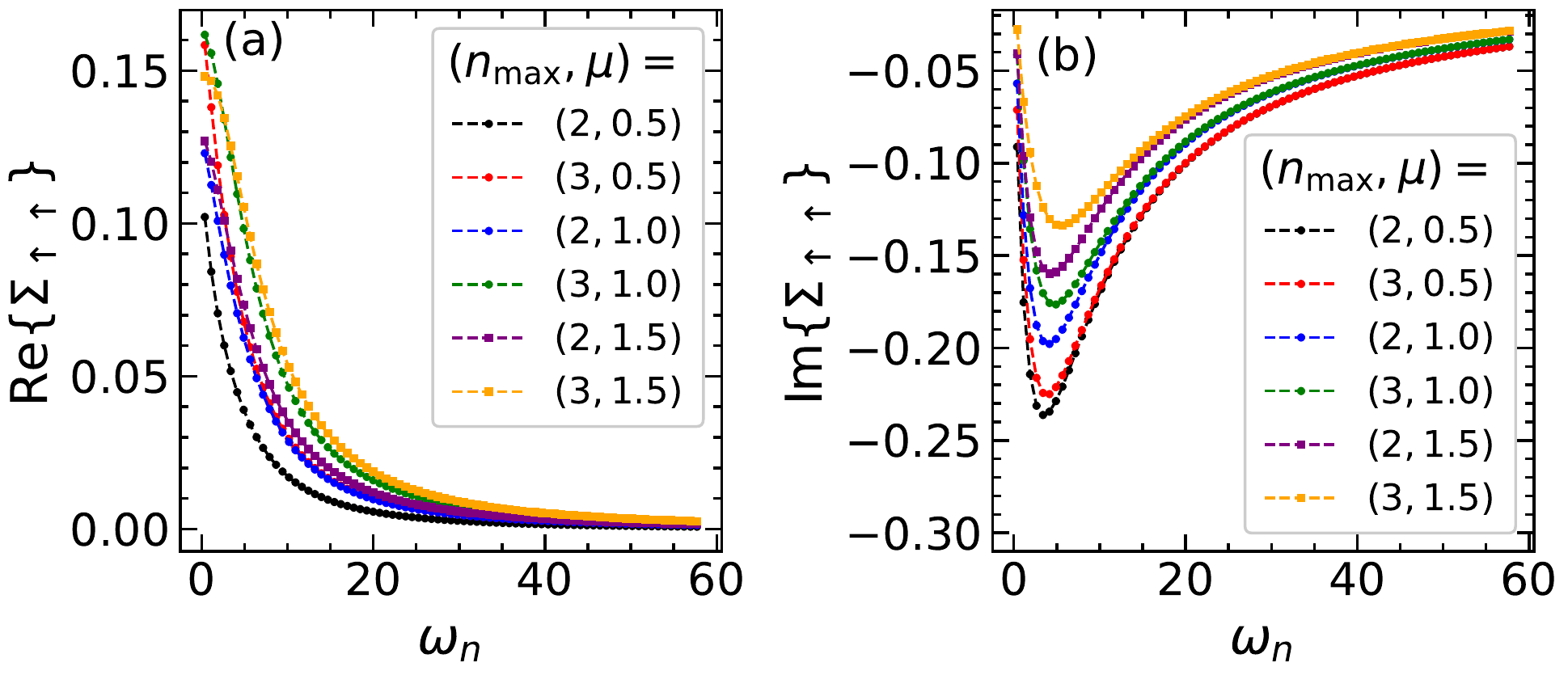}
\caption{(a) The real part of the (spin up) self-energy of the two dimensional Hubbard model for $t=1.0$, $U=3.0$, $\beta=8.33$, $\Vec{k}=(0,\pi)$, and at different values of  $\mu$ as indicated, and (b) are the imaginary counterparts.} 
\label{fig:Hubb_matsubara} 
\end{figure}

\begin{figure}
\centering
    \includegraphics[width=1.0\linewidth]{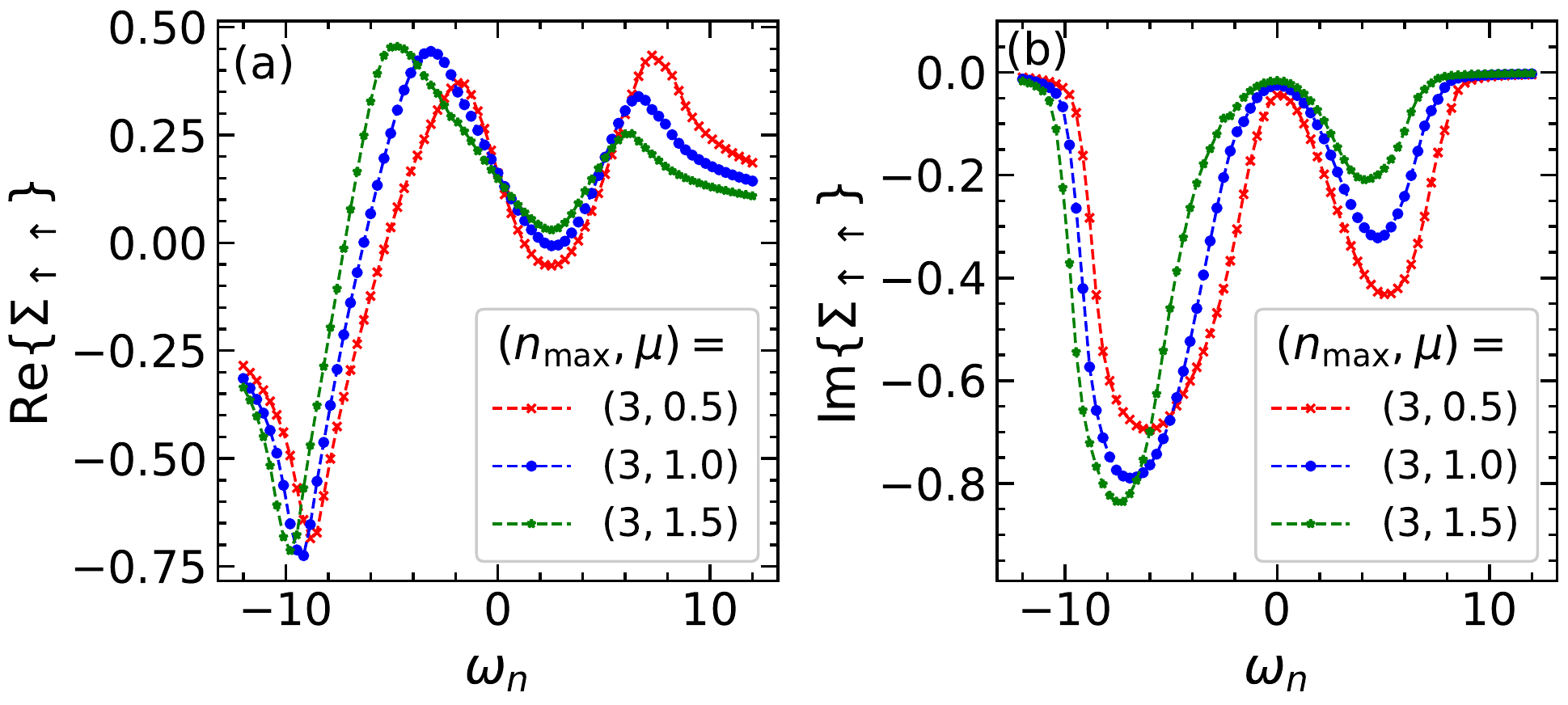}
\caption{(a) and (b) the real and imaginary parts of the self-energy (truncated at third order) versus the real frequency for the 2D Hubbard model evaluated for the parameters choice: $U=3t=3.0$, $\beta=8.33$, and $\Vec{k}=(0,\pi)$ with different values of $\mu$ as indicated. We took a Monte-Carlo sample of size $1\times 10^8$ and the regulator $\Gamma=0.2$.}
\label{fig:Hubb_mu}
\end{figure}

\section{Conclusions}
In this work we have developed an algorithm that can handle single and multiband problems for general two-body interaction models at equilibrium. The steps to our determinant method are: (1) Generating contractions by evaluating the proper determinant, (2) performing the symbolic Fourier transform, (3) using the AMI to evaluate the Matsubara summations exactly, (4) sum or sample any remaining internal degrees of freedom. 

We have applied our algorithm to a variety of problems from molecular chemistry to lattice models up to fourth order perturbation theory. The method is therefore flexible and can solve different models in both real and imaginary frequency domains allowing it to be of great importance for both quantum chemistry and lattice system applications.
The bottleneck in computation of lattice systems remains the numerical integration over remaining spatial degrees of freedom.  When the numerical regulator $\Gamma$ is small this becomes difficult due to the sharp nature of the integrands.  The use of renormalized perturbation theory might help alleviate these difficulties.\cite{burke:2023}
Finally, our algorithm, equivalent to a single shot GFn\cite{hirata:2015} exceeds what is currently available.
Although we limited ourselves to fourth order calculations, higher order corrections can be achieved, since the algorithm is valid at any arbitrary perturbation order and system size.  Of particular interest is molecular problems where we are able to evaluate each perturbative order exactly to machine precision.  In these cases, regardless of the computational expense of higher orders, since the result is exact it need only \emph{ever} be computed once.   

\begin{acknowledgments}
JPFL would like to thank George Booth and Oliver Backhouse for helpful discussion. We would like to thank as well Jia Li and Emmanuel Gull who were instrumental at the onset of this work.
We acknowledge the support of the Natural Sciences and Engineering Research Council of Canada (NSERC) RGPIN-2022-03882 and support from the Simons Collaboration on the Many Electron Problem.
Computational resources were provided by the Digital Research Alliance of Canada. Our Monte Carlo codes make use of the open source ALPSCore framework\cite{ALPSCore,alpscore_v2} and the libami package\cite{libami}.
\end{acknowledgments}

\appendix
\section{Steps of Symbolic Fourier Transformation}
\label{app:SFT}
\subsection{Sorting Wick's contractions}

An important step to perform the symbolic Fourier transformation is to sort the given Wick contraction that corresponds to a connected diagram as follows. First, let's represent the given contraction as $\mathcal{C}=[\Vec{P},s]$ where $\Vec{P}=\left(p_1,p_2,\cdots,p_{2n+1}\right)$ is a vector of pairs representing each Fermion line with $p_j=(\tau^{j}_0,\tau^j_1)$, and $s$ is the sign of the contraction. In the language of graph theory, $\Vec{P}$ contains the edges of the graph. To check if the diagram is connected or not, one can use the Depth First Search (DFS) which requires $\Vec{P}$ as an input. If $\mathcal{C}$ is connected, then we introduce three  vectors of pairs $\Vec{A}$, $\Vec{B}$, and $\Vec{C}$ where we store the pairs from $\Vec{P}$ into these three vectors based on the following convention. The pairs representing connection with external vertices are stored in $\Vec{B}$ and the pairs which representing loops, i.e. tadpole/clamshell structures, are stored in $\Vec{C}$. 

The next step to reduce the number of pairs in $\Vec{A}$ to $n-1$ which is adopted from basic graph theory fact that a given connected graph with $n$ vertices has $n-1$ edges connecting all the vertices together (plus the extra edges). This can be done recursively using the DFS by removing one pair at a time from $\Vec{A}$ and apply the DFS to check if the remaining pairs keeps all the vertices connected or not. If the removal of a given pair doesn't affect the connectedness, then the pair should be added to $\Vec{C}$, otherwise it should be put back into $\Vec{A}$ and then move to the next pair in $\Vec{A}$ and do the same steps until the number of pairs is $n-1$. At this moment, the numbers of pairs in $\Vec{C}$ is $n$, with the total number of pairs in all the three vectors is $2n+1$ as expected. The contraction $\mathcal{C}$ will have the form
\begin{equation}
\label{eqn:contrABC}
    \mathcal{C}=[\Vec{A},\Vec{B},\Vec{C},s]
\end{equation}

\subsection{Array representation of the non-interacting Green's function}
Let us assume that the fermionic line connecting two vertices $\tau_{i}$ and $\tau_{j}$ in an 
$n^{\th}$ order Feynman diagram is represented by a Green's function of the form 
$g(\eta;\tau_{i}-\tau_{j})$ where $\eta$ is a set of quantum labels attached to the 
corresponding Green's function. We introduce the following useful array representation of 
$g(\eta;\tau_{i}-\tau_{j})$  
\begin{equation}
\label{eqn:arrayGreen}
    g(\eta;\tau_{i}-\tau_{j})\vcentcolon=[V_{j}(1-\delta_{ij}),\eta],
\end{equation}
where $V_{j}\in \mathbb{R}^{n}$ is an $n$-dimensional vector defined in the following way:
\begin{itemize}
    \item If the fermionic line connecting two different internal vertices, then $V_{j}$ has $+1$ at the $i^{\rm th}$ row, $-1$ at the $j^{\rm th}$ row, and zeros elsewhere.
    \item $V_{j}$ is the $n$-dimensional zero vector if $\tau_{i}=\tau_{j}$. This is guaranteed by $\delta_{ij}$ in the equation above. 
    \item The two external fermionic lines are represented with $n$-dimensional vector with only one nonzero entry $\pm 1$. Basically, when $\tau_{j}$ external time and $\tau_{i}$ is internal time then $V_{j}$ has entry of $+1$ at the $i^{\rm th}$ row and zeros elsewhere. On the other hand, if $\tau_{i}$ is the external time, and $\tau_{j}$ is an internal time then $V_{j}$ is an entry $-1$ at the $i^{\rm th}$ row and zeros elsewhere.   
\end{itemize}
Following this notation, we can represent a Wick contraction (\ref{eqn:contrABC}) as  
\begin{equation}
\label{eqn:contrM}
    \mathcal{C}=[M,s]
\end{equation}
where $M=(A|B|C)$ is an $n\times 2n+1$ matrix obtained by mapping the pairs in $\{\Vec{A},\Vec{B},\Vec{C}\}$ into columns vectors using the convention explained above. Basically, the $n-1$ pairs in $\Vec{A}$ form an $n\times n-1$ matrix $A$, the 2 pairs in $\Vec{B}$ form an $n\times 2$ matrix B, and the $n$ pairs in $\Vec{C}$ form an $n\times n$ matrix $C$. In the next section, we will use this result to obtain the Fourier transformation of the contraction $\mathcal{C}$.     

\subsection{Symbolic Fourier Transform}

Let us assume that the fermionic lines whose vectors stored in $A$ has the dependent Matsubara frequencies $\{\omega_{1},\omega_{2},\cdots\omega_{n-1}\}$, the ones stored in $B$ has the external frequency $\omega_{\rm ex}$, and the vectors stored in $C$ has the independent Matsubara frequencies $\{\Omega_{1},\Omega_{2},\cdots\Omega_{n}\}$. Defining $\Vec{\Omega}=(\omega_{1},\omega_{2},\cdots,\omega_{n-1},\omega_{\rm ex},\omega_{\rm ex},\Omega_{1},\Omega_{2},\cdots\Omega_{n})^{t}$, then one can show that the equation that connects all the frequencies together is 
\begin{equation}
\label{eqn:M_omega}
    M\Vec{\Omega}=\Vec{0}.
\end{equation}
The above equation is thought of as the set of delta functions which act to enforce conservation laws at each vertex so long as Eq.~\ref{eqn:M_omega} is satisfied. Our task is to represent the dependent frequencies in terms of the other frequencies which is obtained using the above equation, giving 
\begin{equation}
\left(
\begin{array}{ c c c c c c l r }
\omega_{1}\\
\omega_{2}\\
\omega_{3}\\
\vdots\\
\omega_{n-1}\\
\end{array}\right)=\alpha\omega_{\rm ex}+\beta\left(
\begin{array}{ c c c c c c l r }
\Omega_{1}\\
\Omega_{2}\\
\vdots\\
\Omega_{n}\\
\end{array}\right),
\label{eqn:sol}
\end{equation}
where
\begin{equation}
   \alpha=-J^{-1}A^{T}B\left(
\begin{array}{ c c c c c c l r }
1\\
1\\
\end{array}\right), \beta= -J^{-1}A^{T}C
\end{equation}
with $J=A^{T}A$ is an $n-1\times n-1$ matrix. The above Eq.~\ref{eqn:sol} gives a unique representation of the frequency labels which satisfies the conservation laws at all internal vertices. Using this notation, a Green's function with a dependent frequency $\omega_j$, i.e. $g(\eta;\omega_j)$, will be represented as
\begin{equation}
    g_{k}(\eta_{k};\omega_j)=\frac{1}{i\beta_j\cdot\Vec{\Omega}_{\rm ind}+i\alpha_j\omega_{\rm ex}-\varepsilon_{\eta_{k}}}
\end{equation}
where $\beta_j$ is the $jth$ row in $\beta$, $\alpha_j$ is the $jth$ entry in $\alpha$, and $\Vec{\Omega}_{\rm ind}=(\Omega_{1},\Omega_{2},\cdots\Omega_{n})^{t}$. Consequently, we introduce the Fourier transformation of the Wick contraction (\ref{eqn:contrABC}) as
\begin{equation}
    \mathcal{F}\left[\mathcal{C}\right]:=[g_A,g_B,g_C,s],
\end{equation}
where
\begin{equation}
    g_A=[g_{1}(\eta_{1};\omega_1),g_{2}(\eta_{2};\omega_2),\cdots ,g_{n-1}(\eta_{n-1};\omega_{n-1})]
\end{equation}
\begin{equation}
    g_B=[g_{n}(\eta_{n};\omega_{\rm ex}),g_{2}(\eta_{n+1};\omega_{\rm ex})]
\end{equation}
and
\begin{equation}
    g_C=[g_{n+2}(\eta_{n+2};\Omega_1),\cdots,g_{2n+1}(\eta_{2n+1};\Omega_n)]
\end{equation}
where the Fourier transformed Green's functions in $\{g_B,g_C\}$ takes the following simple form
\begin{equation}
    g_{\ell}(\eta_{\ell};\omega)=\frac{1}{i\omega-\varepsilon_{\eta_{\ell}}}
\end{equation}
Finally, the AMI frequency input will simply be
\begin{equation}
\label{eqn:AMIfreq}
    \omega_{\rm AMI}=\left( \begin{array}{c|c}
   \beta &\alpha \\
   \hline
   I_n & 0 \\
\end{array}\right)
\end{equation}
where $I_n$ is an $n\times n$ identity matrix and $0$ here represents an $n$-dimensional zero vector. 
\begin{theorem}
Let $M$ be an $n\times 2n+1$ matrix representing one particular contraction belonging to specific topology $\mathcal{T}$ with $M$ satisfying (\ref{eqn:M_omega}), then the frequency matrix $\omega_{\rm AMI}$ (\ref{eqn:AMIfreq}) is unique for all contractions belonging to the same $\mathcal{T}$. 
\end{theorem}
\begin{proof}
We know that there are $2^n n!$ contractions per topology $\mathcal{T}$ at $n^{\rm th}$ order. The factor $2^n$ coming from inverting the interaction line at each vertex which essential keeps $M$ invariant. The factorial part coming from relabelling the vertices which is equivalent to re-arranging the rows in $M$. Let $P$ be an $n\times n$ orthogonal matrix that permutes the rows in $M$ bringing it to a new matrix $\tilde{M}:=(\tilde{A}|\tilde{B}|\tilde{C})=PM$. This is equivalent to setting $\tilde{A}=PA$, $\tilde{B}=PB$, and $\tilde{C}=PC$. Clearly, $\tilde{J}=\tilde{A}^T\tilde{A}=J$, $\tilde{A}^T\tilde{B}=A^TB$ and $\tilde{A}^T\tilde{C}=A^TC$. Thus, $\tilde{\alpha}=\alpha$ and $\tilde{\beta}=\beta$.
\end{proof}
The frequency labels can be not unique for a given diagram due to the several possible options of our choice of $A$ and equivalently $C$. In graph theory language, this has to do with the existence of several \textit{directed trees} that are consisting of $n-1$ edges connecting the $n$ vertices. Regardless of this starting choice, the above theorem implies that all of the sibling diagrams in the same topology will always have the same frequency labels once the labels are fixed for one diagram (the AMI input matrix Eq. \ref{eqn:AMIfreq}).

\bibliographystyle{apsrev4-1}
\bibliography{refs.bib}

\end{document}